\documentclass{article}
\usepackage[utf8]{inputenc}

\usepackage{fullpage}
\usepackage{amsthm,amsmath,amssymb}
\usepackage{mathtools}
\usepackage{color}
\usepackage[colorlinks,citecolor=blue,bookmarks=true,linktocpage]{hyperref}
\usepackage{aliascnt}
\usepackage[numbered]{bookmark}
\usepackage[backend=biber,url=false,arxiv=abs,maxbibnames=100,isbn=false]{biblatex}
\usepackage[capitalise]{cleveref}

\newcommand{\opt}{{\rm opt}}
\newcommand{\dist}{{\rm dist}}
\newcommand{\NN}{{\mathcal N}}

\newcommand{\Land}{{~\land~}}
\newcommand{\Lor}{{~\lor~}}

\newtheorem{definition}{Definition}
\newtheorem{theorem}{Theorem}
\newtheorem{lemma}{Lemma}

\Crefname{equation}{Equalities}{Equalities}

\title{Minimum Consistent Subset for Trees Revisited\thanks{This work was partially supported by
  JSPS KAKENHI JP18H04091, 
  JP20H00595, 
  JP20H05793, 
  JP21K11752, 
  JP22H00513, 
  JP23H03344, and 
  JP23KJ1066. 
}}
\author{
Hiroki Arimura\thanks{Hokkaido University.
\texttt{\{arim@ist, koba@ist, nocchi0524@eis\}.hokudai.ac.jp}} \and
Tatsuya Gima\thanks{Nagoya University, JSPS Research Fellow 
\texttt{gima@nagoya-u.jp}} \and
Yasuaki Kobayashi\footnotemark[1] \and
Hiroomi Nochide\footnotemark[1] \and
Yota Otachi\thanks{Nagoya University \texttt{otachi@nagoya-u.jp}}}

\begin{document}

\maketitle

\begin{abstract}
    In a vertex-colored graph $G = (V, E)$, a subset $S \subseteq V$ is said to be consistent if every vertex has a nearest neighbor in $S$ with the same color.
    The problem of computing a minimum cardinality consistent subset of a graph is known to be NP-hard.
    On the positive side, Dey et al. (FCT 2021) show that this problem is solvable in polynomial time when input graphs are restricted to bi-colored trees.
    In this paper, we give a polynomial-time algorithm for this problem on $k$-colored trees with fixed $k$.
\end{abstract}

\section{Introduction}
Let $G = (V, E)$ be an undirected graph and let $c\colon V \to [k]$ be a (not necessarily proper) coloring of $V$.
For $u, v \in V$, we denote by $\dist(u, v)$ the shortest path distance between $u$ and $v$ in $G$.
For a vertex $u \in V$ and $S \subseteq V$, the set of \emph{nearest neighbors} from $u$ in $S$, denoted $\NN_G(S, u)$, consists of vertices $v \in S$ such that
\begin{align*}
    \dist(u, v) = \min_{w \in S} \dist(u, w).
\end{align*}
In other words, $\NN_G(S, u)$ is the set of closest vertices in $S$ from $u$.
The distance between $u$ and its nearest neighbors in $S$ is denoted by $\dist(u, S)$ (i.e., $\dist(u, S) = \min_{w \in S}\dist(u, w)$).
For convenience, we define $\dist(u, S) = \infty$ when $S = \emptyset$.
We say that $S \subseteq V$ is a \emph{consistent subset} of $G$ if for every 
$u \in V$, $\NN_G(S, u)$ contains at least one vertex $v$ with $c(u) = c(v)$.
In this paper, we consider the following problem.

\begin{definition}[\textsc{Minimum Consistent Subset}] Given a vertex-colored graph $G = (V, E)$ with $c\colon V \to [k]$, the goal is to compute a minimum cardinality consistent subset of $G$.
\end{definition}

Motivated by applications in pattern recognition, \textsc{Minimum Consistent Subset} is studied under a geometric setting.
In this setting, we are given a set of colored points in the plane and asked to find a minimum size consistent subset of the points, where the consistency is defined in the same way as \textsc{Minimum Consistent Subset}.
Wilfong~\cite{Wilfong:IJCGA:Nearest:1992} showed that this problem is NP-hard for $3$-colored point sets.
Khodamoradi et al.~\cite{KhodamoradiKR:CALDAM:Consistent:2018} improved this hardness result by showing NP-hardness for bi-colored point sets. 
The complexity of \textsc{Minimum Consistent Subset} (on graphs) is discussed in \cite{BanerjeeBC:LATIN:Algorithms:2018}.
They showed that the problem is NP-hard even on bi-colored graphs.
Dey et al.~\cite{caldam:DeyMN:Minimum:2021} studied \textsc{Minimum Consistent Subset} on several simple classes of graphs and showed that the problem is solvable in polynomial time on paths, caterpillars, bi-colored spiders, and bi-colored comb graphs.
Dey et al.~\cite{fct:DeyMN:Minimum:2021} focused on \textsc{Minimum Consistent Subset} for trees.
They gave an $O(n^4)$-time algorithm for bi-colored $n$-vertex trees and left the complexity of \textsc{Minimum Consistent Subset} on trees with more than two colors as an open question.
Very recently, Manna and Roy~\cite{MannaR:arxiv:Some:2023} showed that \textsc{Minimum Consistent Subset} is NP-complete even on trees when the number of colors $k$ is given as part of the input.

In this paper, we complement the results of \cite{fct:DeyMN:Minimum:2021,MannaR:arxiv:Some:2023} by showing that \textsc{Minimum Consistent Subset} is solvable in polynomial time on trees with $k$ colors for any constant $k$.

\begin{theorem}\label{thm:main}
    There is an $O(2^{4k}n^{2k + 3})$-time algorithm for \textsc{Minimum Consistent Subset} on $k$-colored $n$-vertex trees.
\end{theorem}

We would like to mention that for $k = 2$, the algorithm of \cite{fct:DeyMN:Minimum:2021} runs in time $O(n^4)$, which is faster than ours.
They highly exploit the fact that input trees have exactly two colors, while we employ a standard dynamic programming algorithm on trees, which allows us to solve the problem with more than two colors.
We would also like to stress that our algorithm is conceptually simple (as it is dynamic programming) but the recurrence for computing optimal solutions for subproblems is rather involved due to a ``non-locality'' of consistent subsets (see~\Cref{fig:example}).
This would indicate that it is still challenging to extend our algorithm to more general classes of graphs, such as cactus, series-parallel graphs, and bounded-treewidth graphs.
\begin{figure}
    \centering
    \includegraphics{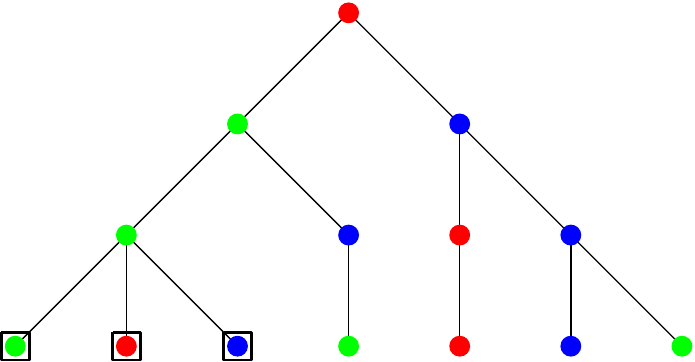}
    \caption{The figure depicts a $3$-colored tree and a consistent subset in it.
    The vertices in the consistent subset are surrounded by boxes.}
    \label{fig:example}
\end{figure}


\section{Preliminaries}
Let $G$ be a graph.
The vertex set and edge set of $G$ are denoted by $V(G)$ and $E(G)$, respectively.
Let $k$ be a positive integer.
We use $[k]$ to denote the set of positive integers not greater than $k$ and $[k]_0^\infty \coloneqq [k] \cup \{0, \infty\}$.
Let $c\colon V(G) \to [k]$ be a (not necessarily proper) vertex coloring of $G$.
In this paper, unless stated otherwise, we consider vertex-colored graphs, and every graph is implicitly associated to some vertex coloring $c$.
For $v \in G$, we denote by $N_G(v)$ the set of neighbors of $v$ in $G$.

Let $S \subseteq V(G)$.
A vertex $v$ is \emph{consistent} (in $(G, S)$) if there is a vertex $w \in \NN_{G}(S, v)$ with $c(w) = c(v)$.
A color $i$ is \emph{consistent} (in $(G, S)$) if every $v \in V(G)$ with $c(v) = i$ is consistent in $(G, S)$.
A vertex or color is \emph{inconsistent} (in $(G, S)$) if it is not consistent in $(G, S)$.
Note that every vertex (and then every color) is considered to be inconsistent when $S$ is empty.
We say that $S$ is \emph{consistent} in $G$ if every color in $[k]$ is consistent in $(G, S)$.


\section{A polynomial-time algorithm for \texorpdfstring{$k$-colored trees with fixed $k$}{}}
Let $T$ be a vertex-colored tree with $c\colon V(T) \to [k]$.
For $1 \le i \le k$, we let $C_i = c^{-1}(i)$, that is, $C_i$ is the set of vertices of $T$ colored in $i$.
We consider $T$ as a rooted tree by taking an arbitrary vertex $r$ as its root.
Let $v \in V(T)$.
We denote by $T_v$ the subtree of $T$ rooted at $v$.
For $\ell \in [n]_0^\infty$, $L, H \subseteq [k]$, and $r_1, \ldots, r_k \in [n]_0^\infty$, a vertex subset $S \subseteq V(T_v)$ is said to be \emph{admissible} for a tuple $t = (\ell, L, H, r_1, \ldots, r_k)$ (or \emph{$t$-admissible}) if it satisfies the following conditions:
\begin{itemize}
    \item $\ell = \displaystyle\dist(v, S)$;
    \item $L = \{i \in [k]: \NN_{T_v}(S, v) \text{ contains a vertex of color } i\}$;
    \item $H = \{i \in [k]: i \text{ is inconsistent in } (T_v, S)\}$, where
    color $i$ is considered to be consistent if it does not appear in $T_v$;
    \item for $1 \le i \le k$,
    \begin{align*}
        r_i = \begin{dcases}
            \min_{w}(\dist(w, S) - \dist(w, v)) & \text{if } i \in H\\
            \max_{w}(\dist(w, S) - \dist(w, v)) & \text{if } i \notin H
        \end{dcases},
    \end{align*}
    where the minimum in the first case is taken over all inconsistent $w \in V(T_v) \cap C_i$ and the maximum in the second case is taken over all (consistent) $w \in V(T_v) \cap C_i$.
    Note that, in the first case, $r_i = \infty$ when $S = \emptyset$, and, in the second case, we define $r_i = 0$ when $V(T_v) \cap C_i = \emptyset$.
    When $r_i < 0$, we define $r_i = 0$ instead.
\end{itemize}
\begin{figure}
    \centering
    \includegraphics{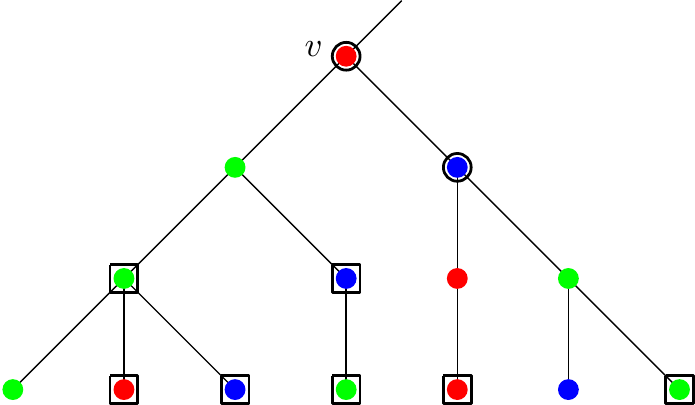}
    \caption{The figure illustrates an example of a $t$-admissible subset $S$ in $T_v$ for $t = (2, \{1, 2\}, \{0, 1\}, 2, 0, 1)$.
    Red, green, and blue are considered as $0$, $1$, and $2$, respectively.
    The vertices in $S$ are indicated by boxes and the inconsistent vertices are indicated by circles. 
    }
    \label{fig:def-opt}
\end{figure}
\Cref{fig:def-opt} illustrates a specific example for a $t$-admissible subset $S$ for some tuple $t$.
Note that for a (connected) subtree $T'$ of $T$ and $u, w \in V(T')$, we have $\dist_{T'}(u, w) = \dist_T(u, w)$, and hence we may simply write $\dist(u, w)$ without specifying subscripts.
We denote by $\opt(v, t)$ the minimum cardinality of a $t$-admissible subset in $T_v$.
If there is no $t$-admissible subset in $T_v$, we define $\opt(v, t) = \infty$. 
Given this, a $t$-admissible vertex subset $S \subseteq V(T)$ is a consistent subset of $T$ for $t = (\ell, L, H, r_1, \ldots r_k)$ if and only if $H = \emptyset$, that is, all colors in $[k]$ are consistent in $(T, S)$.
Thus, $\min_t\opt(r, t)$ is the minimum cardinality of a consistent subset of $T$, where the minimum is taken over all tuples $t = (\ell, L, H, r_1, \ldots r_k)$ with $H = \emptyset$.

The idea behind the above (slightly less intuitive) definition and our dynamic programming is as follows.
For each $v \in V(T)$ and each possible tuple $t = (\ell, L, H, r_1, \ldots, r_k)$, we compute $\opt(v, t)$ in a bottom-up manner.
Let $S \subseteq V(T_v)$ that is admissible for $t$. 
To construct a consistent subset $S^*$ of the whole tree $T$ with $S^* \cap V(T_v) = S$, we need to carefully take the consistency between $S$ and $S^* \setminus S$ into account.

First, each vertex in $S$ may become a nearest neighbor in $S^*$ for a vertex $w \in V(T) \setminus V(T_v)$.
This nearest neighbor in $S$ must be a closest vertex from $v$, that is, it belongs to $\NN_{T_v}(S, v)$.
The values of $\ell$ and $L$ give sufficient information to $w$ to have such a nearest neighbor within $T_v$.
Next, suppose that $T_v$ has a vertex $w \in C_i$ such that $\NN_{T_v}(S, w)$ has no vertex of color $i$, that is, vertex $w$ (and color $i$) is inconsistent in $(T_v, S)$.
By the definition of $H$, we have $i \in H$.
To ensure the consistency of $S^*$, $S^* \setminus S$ must contain a vertex $x \in V(T) \setminus V(T_v)$ of color $i$ such that $\dist(w, x) \le \dist(w, S)$.
This implies that
\begin{align*}
    \dist(v, x) = \dist(w, x) - \dist(w, v) \le \dist(w, S) - \dist(w, v).  
\end{align*}
The values of $r_i$ and $H$ convey this information to the ``outside'' of $T_v$.
Finally, let $w$ be a vertex of color $i$ in $T_v$.
Suppose that $w$ is consistent in $(T_v, S)$, that is, $w$ has a nearest neighbor $x$ of the same color $i$.
If $S^* \setminus S$ has a vertex $y$ of color $j \neq i$ that is strictly closer than $x$ to $w$, the vertex $y$ may destroy the consistency of $S^*$.
More precisely, suppose that $w \in V(T_v) \cap C_i$ has a nearest neighbor $x \in \NN_{T_v}(S, w)$ of color $i$.
Then, $S^*$ cannot contain $y \in V(T) \setminus V(T_v)$ of color $j \neq i$ with $\dist(w, y) < \dist(w, x)$ unless $S^*$ includes another vertex $z$ of color $i$ with $\dist(w, z) \le \dist(w, y)$.
This implies that any vertex $y$ of color $j \neq i$ in $S^* \setminus S$ satisfies
\begin{align*}
    \dist(v, y) = \dist(w, y) - \dist(w, v) \ge \dist(w, S \cap C_i) - \dist(w, v) = \dist(w, S) - \dist(w, v)
\end{align*}
unless $S^* \setminus S$ includes $z$ of color $i$ with $\dist(w, z) \le \dist(w, y)$.
Let us note that the last equality follows from the fact that $i \notin H$.
The values of $r_i$ and $H$ will be used for preventing $S^*$ from including such a vertex $y$.

Now, we describe an algorithm to compute $\opt(v, t)$ for $v \in V(T)$ and tuple $t = (\ell, L, H, r_1, \ldots, r_k)$ in a bottom up manner.
We distinguish the following two cases.

\paragraph{Leaf case.}
Suppose that $v$ is a leaf of $T$.
Then, there are two possible vertex subsets $\emptyset$ and $\{v\}$ in $T_v$.
The following equality follows from the definition.
\begin{align*}
    \opt(v, t) = \begin{dcases}
        0 & \text{if }\ell = \infty \Land L = \emptyset \Land H = \{c(v)\} \Land r_{c(v)} = \infty \Land r_i = 0\ (\forall i \in [k] \setminus \{c(v)\})\\
        1 & \text{if } \ell = 0 \Land L = \{c(v)\} \Land H = \emptyset \Land r_1 = \cdots = r_k = 0\\
        \infty & \text{otherwise}
    \end{dcases}.
\end{align*}

\paragraph{Internal vertex case.}
Suppose that $v$ is an internal vertex of $T$.
Let $v_1, \ldots, v_q$ be the children of $v$ in $T$.
We extend the notation of $T_v$ as follows.
We denote by $T_v^0$ the subtree of $T$ containing exactly one vertex $v$, and for $1 \le j \le q$, by $T_v^j$ the subtree of $T$ consisting of all vertices in $T_v^{j - 1}$ and $T_{v_j}$, that is, $T_v^j$ is obtained from $T_v$ by deleting all vertices in $T_{v_{j'}}$ for all $j' > j$.
We also extend the definition of $\opt$ as follows.
The definition of $t$-admissibility can be extended for $T_{v}^j$ by just replacing $T_v$ with $T_{v}^j$.
We define $\opt(v, j, t)$ as the minimum cardinality of a $t$-admissible subset of $T_{v}^j$.
Clearly, we have $\opt(v, t)  = \opt(v, q, t)$.
Since $T_{v}^0$ consists of the single vertex $v$, $\opt(v, j, t)$ for $j = 0$ can be computed as in the leaf case.
Thus, in the following, we consider the other case $j \ge 1$.

Let $j \ge 1$.
For the sake of notational convenience, we write $T^*$, $T'$, and $T''$ to denote $T^j_v$, $T^{j-1}_{v}$ and $T_{v_j}$, respectively. (See \Cref{fig:LR-tree}.)
\begin{figure}
    \centering
    \includegraphics[width=0.25\textwidth]{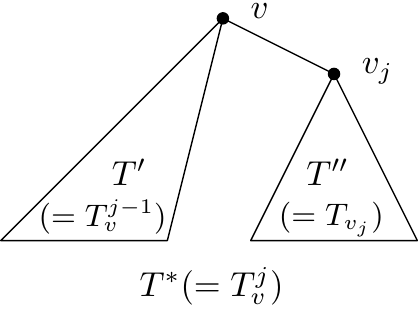}
    \caption{The figure illustrates subtrees $T'$ and $T''$ in $T^*$.}
    \label{fig:LR-tree}
\end{figure}
Let $t' = (\ell', L', H', r'_1, \ldots, r'_k)$ and $t'' = (\ell'', L'', H'', r''_1, \ldots, r''_k)$ be two tuples, and let $S'$ and $S''$ be $t'$-admissible and $t''$-admissible subsets of $T'$ and $T''$, respectively.
Let $S = S' \cup S''$.
Note that $S' \cap S'' = \emptyset$.
In the following, we consider a tuple $t = (\ell, L, H, r_1, \ldots, r_k)$ for which $S$ is admissible in $T^*$.

The values of $\ell$ and $L$ are easily determined from $\ell', \ell'', L', L''$. Since $\ell' = \dist(v, S')$ and $\ell''+1 = \dist(v_j, S'') + 1 = \dist(v, S'')$,
we have 
\begin{align}\label{exp:ell}
    \ell = \dist(v, S) = \min(\dist(v, S'), \dist(v, S'')) = \min(\ell', \ell'' + 1).
\end{align}
Moreover, $L$ is defined as:
\begin{align}\label{exp:L}
    L = \begin{dcases}
       L' &  \text{if }\ell' < \ell'' + 1\\
       L'' & \text{if } \ell' > \ell''+ 1\\
       L' \cup L'' &\text{if } \ell' = \ell'' + 1
    \end{dcases}.
\end{align}

To see the values of $H$ and $r_i$ for $i \in [k]$, we further distinguish two cases.
In the following, we fix a color $i \in [k]$.

\paragraph{Case I: color $i \in [k]$ is consistent in $(T^*, S)$.}
Suppose that color $i \in [k]$ is consistent in $(T^*, S)$.
In this case, every vertex $w \in V(T^*) \cap C_i$ has a vertex $x \in \NN_{T^*}(S, w)$ with $c(x) = i$.

Suppose first that color $i$ is consistent in $(T', S')$, that is, every $w \in V(T') \cap C_i$ is consistent in $(T', S')$.
If $\dist(w, S') \le \dist(w, S'')$ for every $w \in V(T') \cap C_i$, then all vertices in $V(T') \cap C_i$ are indeed consistent in $(T^*, S)$ as $\NN_{T'}(S', w) \subseteq \NN_{T^*}(S, w)$.
Note that the condition $\dist(w, S') \le \dist(w, S'')$ can be equivalently transformed as:
\begin{align*}
    \dist(w, S') \le \dist(w, S'') &\iff \dist(w, S') - \dist(w, v) \le \dist(w, S'') - \dist(w, v)\\
    &\iff \dist(w, S') - \dist(w, v) \le \ell'' + 1.
\end{align*}
Since $r'_i = \max_{w \in V(T') \cap C_i} (\dist(w, S') - \dist(w, v))$, the above condition is simply written as $r'_i \le \ell'' + 1$.
Otherwise, there is $w \in V(T') \cap C_i$ such that $\dist(w, S') > \dist(w, S'')$.
As $\NN_{T^*}(S, w) = \NN_{T''}(S'', v_j)$, $L''$ must contain color $i$.
Hence, these two cases require that at least one of $r'_i \le \ell'' + 1$ or $i \in L''$ holds when $i \notin H'$.

Symmetrically, suppose that color $i$ is consistent in $(T'', S'')$, that is, every vertex $w \in V(T'') \cap C_i$ is consistent in $(T'', S'')$.
Then, at least one of $\dist(w, S') \ge \dist(w, S'')$ for every $w \in V(T'') \cap C_i$ or $i \in L'$ is satisfied.
The first condition is equivalent to $r''_i \le \ell' + 1$ as
\begin{align*}
    \dist(w, S'') \le \dist(w, S') &\iff \dist(w, S'') - \dist(w, v_j) - 1 \le \dist(w, S') - \dist(w, v_j) - 1\\
    &\iff \dist(w, S'') - \dist(w, v_j) \le  \ell' + 1.
\end{align*}
Thus, $r''_i \le \ell' + 1$ or $i \in L'$ is satisfied when $i \notin H''$.

Suppose next that color $i$ is inconsistent in $(T', S')$, that is, there is a vertex $w \in V(T') \cap C_i$ that is inconsistent in $(T', S')$.
Then, to ensure that color $i$ is consistent in $(T^*, S)$, there must be a nearest neighbor $x \in S''$ of color $i$ such that $\dist(w, S') \ge \dist(w, x)$.
Moreover, such a vertex $x$ belongs to $\NN_{T'}(S'', v_j)$, which means that $i \in L''$.
The above inequality is equivalent to
\begin{align*}
   \dist(w, S') \ge \dist(w, x) &\iff \dist(w, S') - \dist(w, v) \ge \dist(w, x) - \dist(w, v)\\
   &\iff \dist(w, S') - \dist(w, v) \ge \ell'' + 1.
\end{align*}
This condition must hold for all $w \in V(T') \cap C_i$ that are inconsistent in $(T', S')$.
Thus, we have $r'_i \ge \ell'' + 1$ when $i \in H'$.
Symmetrically, suppose that color $i$ is inconsistent in $(T'', S'')$, that is, there is a vertex $w$ that is inconsistent in $(T'', S'')$.
Then, we have $i \in L'$ and $r''_i \ge \ell' + 1$ when $i \in H''$.

To summarize, we say that color $i$ is \emph{left consistent} (for $(T^*, S)$) if one of the following conditions holds:
\begin{itemize}
    \item $i \notin H' \Land (r'_i \le \ell'' + 1 \Lor i \in L'')$;
    \item $i \in H' \Land i \in L'' \Land r'_i \ge \ell'' + 1$.
\end{itemize}
Similarly, we say that color $i$ is \emph{right consistent} (for $(T^*, S)$) if one of the following conditions holds:
\begin{itemize}
    \item $i \notin H'' \Land (r''_i \le \ell' + 1 \Lor i \in L')$;
    \item $i \in H'' \Land i \in L' \Land r''_i \ge \ell' + 1$.
\end{itemize}

The following lemma summarizes the above observations.
\begin{lemma}\label{lem:case1}
    A color $i$ is consistent in $(T^*, S)$ if it is left and right consistent.
\end{lemma}
\begin{proof}
    By symmetry, it suffices to show that every vertex in $w \in V(T') \cap C_i$ is consistent in $(T^*, S)$ when color $i$ is left consistent.
    Suppose $i \notin H'$.
    If $i \in L''$, then both $\NN_{T}(S', w)$ and $\NN_{T}(S'', w)$ contains vertices of color $i$.
    Hence, $i$ is consistent in $(T^*, S)$.
    Otherwise, we have $r'_i \le \ell'' + 1 $.
    Then, for $w \in V(T') \cap C_i$ and $x \in S''$,
    \begin{align*}
        \dist(w, S') \le r'_i + \dist(w, v) \le \ell'' + 1 + \dist(w, v) \le \dist(w, x).
    \end{align*}
    This implies that $w$ is consistent in $(T^*, S)$ as $i$ is consistent in $(T', S')$.
    
    Suppose $i \in H'$.
    Let $w \in V(T') \cap C_i$ be a vertex that is inconsistent in $(T', S')$.
    Since $i \in L''$, there is a vertex $x \in \NN_{T''}(S'', v_j)$ of color $i$.
    As $r'_i \ge \ell'' + 1$, we have
    \begin{align*}
        \dist(w, S') \ge r_i + \dist(w, v) \ge \ell'' + 1 + \dist(w, v) = \dist(w, x).
    \end{align*}
    Thus, we have $x \in \NN_{T}(S, w)$, implying that $w$ is consistent in $(T^*, S)$.
\end{proof}

Next, we turn to considering the value of $r_i$ when $i$ is consistent in $(T^*, S)$.
The value $r_i$ is decomposed as
\begin{align*}
    r_i &= \max_{w \in V(T^*) \cap C_i} (\dist(w, S) - \dist(w, v))\\
    &= \max\left( \max_{w \in V(T') \cap C_i}( \dist(w, S) - \dist(w, v)), \max_{w \in V(T'') \cap C_i} (\dist(w, S) - \dist(w, v)) \right).
\end{align*}
Let $r_i^{\rm left} = \max_{w \in V(T') \cap C_i}(\dist(w, S) - \dist(w, v))$ be the first term in the above equality and let $w^* \in V(T') \cap C_i$ such that $r_i^{\rm left} = \dist(w^*, S) - \dist(w^*, v)$.
Suppose that $i \notin H'$ and $r'_i \le \ell'' + 1$.
As $r'_i \le \ell'' + 1$, we have $\dist(w^*, S') \le \dist(w^*, x)$ for any $x \in S''$, implying that $\dist(w^*, S) = \dist(w^*, S')$.
Thus, we have $r_i^{\rm left} = r'_i$.
Otherwise, there are two cases: (I-1) $i \notin H'$ and $r'_i > \ell'' + 1$ or (I-2) $i \in H'$ and $r'_i \ge \ell'' + 1$.
In both cases, there is a vertex $w \in V(T') \cap C_i$ such that $\dist(w, S') \ge \dist(w, S'')$.
Then, we have
\begin{align*}
    \dist(w, S) - \dist(w, v) = \dist(w, S'') - \dist(w, v) = \dist(w, v) + \ell'' + 1 - \dist(w, v) = \ell'' + 1.
\end{align*}
To compute $r^{\rm left}_i$, the following lemma comes in handy.

\begin{lemma}\label{lem:case1:update}
    Let $w_1, w_2 \in V(T') \cap C_i$ such that $\dist(w_1, S) - \dist(w_1, v) \ge \dist(w_2, S) - \dist(w_2, v)$.
    Suppose that $\dist(w_2, S') \ge \dist(w_2, S'')$.
    Then,
    \begin{align*}
        \dist(w_1, S) - \dist(w_1, v) = \dist(w_2, S) - \dist(w_2, v) = \ell'' + 1.
    \end{align*}
\end{lemma}
\begin{proof}
    As $\dist(w_2, S') \ge \dist(w_2, S'')$, we have $\dist(w_2, S) - \dist(w_2, v) = \ell'' + 1$.
    If $\dist(w_1, S') \ge \dist(w_1, S'')$, we also have $\dist(w_1, S) - \dist(w_1, v) = \ell' + 1$ and then we are done.
    Suppose for contradiction that $\dist(w_1, S') < \dist(w_1, S'')$.
    This implies that
    \begin{align*}
        \dist(w_1, S) - \dist(w_1, v) = \dist(w_1, S') - \dist(w_1, v) < \dist(w_1, S'') - \dist(w_1, v) = \ell'' + 1.
    \end{align*}
    However, as $\dist(w_1, S) - \dist(w_1, v) \ge \dist(w_2, S) - \dist(w_2, v)$,
    \begin{align*}
        \dist(w_1, S) - \dist(w_1, v) \ge \dist(w_2, S) - \dist(w_2, v) = \ell'' + 1,
    \end{align*}
    which yields a contradiction.
\end{proof}

As $\dist(w^*, S) - \dist(w^*, v) \ge \dist(w, S) - \dist(w, v)$, by~\Cref{lem:case1:update}, we have
\begin{align*}
    r_i^{\rm left} = \dist(w^*, S) - \dist(w^*, v) = \dist(w, S) - \dist(w, v) = \ell'' + 1
\end{align*}
in both cases (I-1) and (I-2).
Therefore, by symmetrically considering $r^{\rm right}_i \coloneqq \max_{w \in V(T'') \cap C_i}(\dist(w, S) - \dist(w, v))$, it holds that
\begin{align}\label{exp:case1:r_i}
    r_i = \max(r^{\rm left}_i, r^{\rm right}_i),
\end{align}
where $r^{\rm left}_i = \min(r'_i, \ell'' + 1)$ and $r^{\rm right}_i = \min(\ell'+1, r''_i)$ under the assumption that $i \notin H$.

\paragraph{Case II: color $i \in [k]$ is inconsistent in $(T^*, S)$.}
The previous case (Case I) gives sufficient conditions to ensure the consistency of color $i$.
In fact, these conditions are necessary for the consistency.
We say that color $i$ is \emph{left inconsistent} (for $(T^*, S)$) if $i$ is not left consistent, that is, 
\begin{itemize}
    \item $i \notin H' \implies (r'_i > \ell'' + 1 \Land i \notin L'')$;
    \item $i \in H' \implies (i \notin L'' \Lor r'_i < \ell'' + 1)$.
\end{itemize}
Similarly, we say that color $i$ is \emph{right inconsistent} (for $(T^*, S)$) if $i$ is not right consistent, that is,
\begin{itemize}
    \item $i \notin H'' \implies (r''_i > \ell' + 1 \Land i \notin L')$;
    \item $i \in H'' \implies (i \notin L' \Lor r''_i < \ell' + 1)$.
\end{itemize}

\begin{lemma}\label{lem:case2}
    A color $i$ is inconsistent in $(T^*, S)$ if it is left or right inconsistent.
\end{lemma}
\begin{proof}
    It suffices to show that there is a vertex $w \in V(T') \cap C_i$ that is inconsistent in $(T^*, S)$ if color $i$ is left inconsistent.
    
    Suppose $i \notin H'$.
    As $r'_i > \ell'' + 1$, there is a vertex $w \in V(T') \cap C_i$ such that
    \begin{align*}
        \dist(w, S') = r'_i + \dist(w, v) > \ell'' + 1 + \dist(w, v) = \dist(w, x)
    \end{align*}
    for some $x \in S''$.
    Note that $r'_i > \ell'' + 1$ implies that $\ell''$ takes a finite value and hence $S''$ is nonempty.
    Since $i \notin L''$, we cannot choose such a vertex $x \in \NN_{T''}(S'', v_j)$ of color $i$.
    Thus, $w$ is inconsistent in $(T^*, S)$.

    Suppose $i \in H'$.
    If $i \notin L''$, both $\NN_{T^*}(S', w)$ and $\NN_{T^*}(S'', w)$ have no vertices of color $i$, and hence $i$ is inconsistent in $(T^*, S)$.
    Otherwise, we have $r'_i < \ell'' + 1$.
    This implies that
    \begin{align*}
        \dist(w, S')  = r'_i + \dist(w, v) < \ell'' + 1 + \dist(w, v) = \dist(w, S'')
    \end{align*}
    for some $w \in V(T') \cap C_i$ that is inconsistent in $(T', S')$.
    Thus, $\dist(w, S') = \dist(w, S)$ and hence $\NN_{T^*}(S, w) = \NN_{T^*}(S', w)$, which yields that $i$ is inconsistent in $(T^*, S)$.
\end{proof}

Finally, we turn to considering the value of $r_i$ when $i$ is inconsistent in $(T^*, S)$.
The value $r_i$ is decomposed as
\begin{align*}
        r_i &= \min_{w \in V(T^*) \cap C_i} (\dist(w, S) - \dist(w, v))\\
    &= \min\left( \min_{w \in V(T') \cap C_i}( \dist(w, S) - \dist(w, v)), \min_{w \in V(T'') \cap C_i} (\dist(w, S) - \dist(w, v)) \right),
\end{align*}
where all the minimum values in the above equality are taken among all inconsistent vertices $w$.
Let $r_i^{\rm left} = \min_{w \in V(T') \cap C_i}(\dist(w, S) - \dist(w, v))$ be the first term in the above equality and let $w^* \in V(T') \cap C_i$ be inconsistent such that $r_i^{\rm left} = \dist(w^*, S) - \dist(w^*, v)$.
Suppose $i \in H'$ and $r'_i < \ell'' + 1$.
As $r'_i < \ell'' + 1$, we have $\dist(w^*, S) < \dist(w^*, S'')$, meaning that $\dist(w^*, S) = \dist(w^*, S')$.
Thus, $r^{\rm left}_i = r'_i$.
Suppose otherwise.
There are two cases: (II-1) $i \in H'$, $r'_i \ge \ell'' + 1$, and $i \notin L''$ or (II-2) $i \notin H'$, $r'_i > \ell'' + 1$, and $i \notin L''$.
As $r'_i \ge \ell'' + 1$, there is $w \in V(T') \cap C_i$ such that
\begin{align*}
    \dist(w, S') = r'_i + \dist(w, v) \ge \ell'' + 1 + \dist(w, v) = \dist(w, x)
\end{align*}
for some $x \in \NN_{T^*}(S'', w)$, which means that $\dist(w, S') \ge \dist(w, S'')$.
Moreover, this vertex $w$ is inconsistent in $(T^*, S)$ since (II-1) it is inconsistent in $(T', S'')$ and we cannot choose such $x$ of color $i$ or (II-2) $\NN_{T^*}(S, w) = \NN_{T^*}(S'', w) = \NN_{T^*}(S'', v)$ has no vertices of color $i$.

The following lemma is a ``reverse'' counterpart of \Cref{lem:case1:update}.

\begin{lemma}\label{lem:case2:update}
    Let $w_1, w_2 \in V(T') \cap C_i$ such that $\dist(w_1, S) - \dist(w_1, v) \ge \dist(w_2, S) - \dist(w_2, v)$.
    Suppose that $\dist(w_1, S') \ge \dist(w_1, S'')$.
    Then,
    \begin{align*}
        \dist(w_1, S) - \dist(w_1, v) = \dist(w_2, S) - \dist(w_2, v) = \ell'' + 1.
    \end{align*}
\end{lemma}
\begin{proof}
    As $\dist(w_1, S') \ge \dist(w_1, S'')$, we have $\dist(w_1, S) - \dist(w_1, v) = \ell'' + 1$.
    If $\dist(w_2, S') \ge \dist(w_2, S'')$, we also have $\dist(w_2, S) - \dist(w_2, v) = \ell' + 1$ and then we are done.
    Suppose for contradiction that $\dist(w_2, S') < \dist(w_2, S'')$.
    This implies that
    \begin{align*}
        \dist(w_2, S) - \dist(w_2, v) = \dist(w_2, S') - \dist(w_2, v) < \dist(w_2, S'') - \dist(w_2, v) = \ell'' + 1.
    \end{align*}
    However, as $\dist(w_1, S) - \dist(w_1, v) \ge \dist(w_2, S) - \dist(w_2, v)$,
    \begin{align*}
        \dist(w_2, S) - \dist(w_2, v) \ge \dist(w_1, S) - \dist(w_1, v) = \ell'' + 1,
    \end{align*}
    which yields a contradiction.
\end{proof}

Since $w$ and $w^*$ are both inconsistent in $(T^*, S)$, we have $\dist(w^*, S) - \dist(w^*, v) \le \dist(w, S) - \dist(w, v)$.
By~\Cref{lem:case2:update}, we have
\begin{align*}
    r^{\rm left}_i = \dist(w^*, S) - \dist(w^*, v) = \dist(w, S) - \dist(w, v) = \ell'' + 1
\end{align*}
in both cases (II-1) and (II-2).

Therefore, by symmetrically considering $r^{\rm right}_i \coloneqq \min_{w \in V(T'') \cap C_i}(\dist(w, S) - \dist(w, v))$, it holds that
\begin{align}\label{exp:case2:r_i}
    r_i = \min(r^{\rm left}_i, r^{\rm right}_i),
\end{align}
where $r^{\rm left}_i = \min(r'_i, \ell'' + 1)$ and $r^{\rm right}_i = \min(\ell'+1, r''_i)$ under the assumption that $i \notin H$.

Now, we are ready to describe our recurrence for computing $\opt(v, j, t)$.
We say that a tuple $t = (\ell, L, H, r_1, \ldots, r_k)$ is \emph{compatible} with a pair of tuples $(t', t'')$, where $t' = (\ell', L', H', r'_1, \ldots, r'_k)$ and $t'' = (\ell'', L'', H'', r''_1, \ldots, r''_k)$, if $t$, $t'$, and $t''$ satisfy \Cref{exp:ell,exp:L,lem:case1,exp:case1:r_i,lem:case2,exp:case2:r_i}.
The correctness of the recurrence follows from \Cref{exp:ell,exp:L,lem:case1,exp:case1:r_i,lem:case2,exp:case2:r_i}.

\begin{lemma}\label{lem:rec}
    For $j \ge 1$ and tuple $t$,
    \begin{align*}
        \opt(v, j, t) = \min_{(t', t'')} (\opt(v, j-1, t') + \opt(v_j, t'')),
    \end{align*}
    where the minimum is taken over all pairs $(t', t'')$ such that $t$ is compatible with $(t', t'')$.
\end{lemma}

Finally, we discuss the running time of evaluating the recurrence in \Cref{lem:rec}.
For each $v$ and $1 \le j \le q$, we can evaluate $\opt(v, j, t)$ for all tuple $t$ in time $O((n \cdot 2^k\cdot 2^k \cdot n^{k})^2) = O(2^{4k}n^{2k + 2})$ in total.
Therefore, we can evaluate $\opt(v, j, t)$ for all $v$, $j$, and $t$ in total time $O(2^{4k}n^{2k + 3})$.
From the values of $\opt(v, j, t)$, we can construct a minimum cardinality consistent subset of $T$ by a standard traceback technique with running time in $O(2^{4k}n^{k+3})$, completing the proof of \Cref{thm:main}.

\printbibliography

@inproceedings{caldam:DeyMN:Minimum:2021,
  author    = {Sanjana Dey and
               Anil Maheshwari and
               Subhas C. Nandy},
  title     = {Minimum Consistent Subset of Simple Graph Classes},
  booktitle = {Algorithms and Discrete Applied Mathematics - 7th International Conference,
               {CALDAM} 2021, Rupnagar, India, February 11-13, 2021, Proceedings},
  pages     = {471--484},
  year      = {2021},
  doi       = {10.1007/978-3-030-67899-9_37},
  bibsource = {dblp computer science bibliography, https://dblp.org}
}

@inproceedings{fct:DeyMN:Minimum:2021,
  author    = {Sanjana Dey and
               Anil Maheshwari and
               Subhas C. Nandy},
  title     = {Minimum Consistent Subset Problem for Trees},
  booktitle = {Fundamentals of Computation Theory - 23rd International Symposium,
               {FCT} 2021, Athens, Greece, September 12-15, 2021, Proceedings},
  pages     = {204--216},
  year      = {2021},
  doi       = {10.1007/978-3-030-86593-1_14},
  bibsource = {dblp computer science bibliography, https://dblp.org}
}

@article{Wilfong:IJCGA:Nearest:1992,
  author    = {Gordon T. Wilfong},
  title     = {Nearest neighbor problems},
  journal   = {Int. J. Comput. Geom. Appl.},
  volume    = {2},
  number    = {4},
  pages     = {383--416},
  year      = {1992},
  doi       = {10.1142/S0218195992000226},
  bibsource = {dblp computer science bibliography, https://dblp.org}
}

@inproceedings{KhodamoradiKR:CALDAM:Consistent:2018,
  author    = {Kamyar Khodamoradi and
               Ramesh Krishnamurti and
               Bodhayan Roy},
  editor    = {B. S. Panda and
               Partha P. Goswami},
  title     = {Consistent Subset Problem with Two Labels},
  booktitle = {Algorithms and Discrete Applied Mathematics - 4th International Conference,
               {CALDAM} 2018, Guwahati, India, February 15-17, 2018, Proceedings},
  series    = {Lecture Notes in Computer Science},
  volume    = {10743},
  pages     = {131--142},
  publisher = {Springer},
  year      = {2018},
  doi       = {10.1007/978-3-319-74180-2_11},
  bibsource = {dblp computer science bibliography, https://dblp.org}
}

@inproceedings{BanerjeeBC:LATIN:Algorithms:2018,
  author    = {Sandip Banerjee and
               Sujoy Bhore and
               Rajesh Chitnis},
  editor    = {Michael A. Bender and
               Martin Farach{-}Colton and
               Miguel A. Mosteiro},
  title     = {Algorithms and Hardness Results for Nearest Neighbor Problems in Bicolored
               Point Sets},
  booktitle = {{LATIN} 2018: Theoretical Informatics - 13th Latin American Symposium,
               Buenos Aires, Argentina, April 16-19, 2018, Proceedings},
  series    = {Lecture Notes in Computer Science},
  volume    = {10807},
  pages     = {80--93},
  publisher = {Springer},
  year      = {2018},
  doi       = {10.1007/978-3-319-77404-6_7},
  bibsource = {dblp computer science bibliography, https://dblp.org}
}

@article{MannaR:arxiv:Some:2023,
  author       = {Bubai Manna and
                  Bodhayan Roy},
  title        = {Some results on Minimum Consistent Subsets of Trees},
  journal      = {CoRR},
  volume       = {abs/2303.02337},
  year         = {2023},
  doi          = {10.48550/arXiv.2303.02337},
  eprinttype    = {arXiv},
  eprint       = {2303.02337},
  bibsource    = {dblp computer science bibliography, https://dblp.org}
}
\end{document}